\documentclass[a4paper,11pt]{article}
\usepackage[T1]{fontenc}

\usepackage{amsthm,thmtools}
\usepackage{a4wide}
\usepackage{graphicx}
\usepackage{esvect} 
\usepackage{tcolorbox}
\usepackage{complexity}
\usepackage{tikz}
\usepackage{amsmath}
\usepackage{cleveref}
\usepackage{mathrsfs}
\usepackage{multibib}
\usepackage{enumerate}
\usepackage{authblk}
\usepackage{lineno} 
\newcites{app}{Additional references in the appendix}

\usetikzlibrary{graphs}

\declaretheorem[style=plain,name=Theorem]{theorem}
\declaretheorem[style=plain,sibling=theorem,name=Lemma]{lemma}
\declaretheorem[style=plain,sibling=theorem,name=Proposition]{proposition}
\declaretheorem[style=plain,sibling=theorem,name=Corollary]{corollary}

\declaretheorem[style=plain,sibling=theorem,name=Conjecture]{conjecture}

\crefname{theorem}{Theorem}{Theorems}
\crefname{proposition}{Proposition}{Propositions}
\crefname{lemma}{Lemma}{Lemmas}
\crefname{exmp}{Example}{Examples}
\crefname{corollary}{Corollary}{Corollarys}
\crefname{claim}{Claim}{Claims}
\crefname{remark}{Remark}{Remarks}
\crefname{section}{Section}{Sections}
\newcommand{\prb}[1]{\textnormal{\scshape #1}}

\newcommand{\Rule}{\mathsf{R}}
\newcommand{\kTJ}{k\text{-}\mathsf{TJ}}
\newcommand{\numTJ}[1]{$#1$\text{-}\mathsf{TJ}}
\newcommand{\TJ}{\mathsf{TJ}}
\newcommand{\kTS}{k\text{-}\mathsf{TS}}
\newcommand{\numTS}[1]{$#1$\text{-}\mathsf{TS}}
\newcommand{\TS}{\mathsf{TS}}

\newcommand{\bicliquesize}{b}

\newcommand{\hISIS}{$H$-induced subgraph isomorphic set}
\newcommand{\gISIS}[1]{$#1$-induced subgraph isomorphic set}
\newcommand{\hgraph}{pattern graph}
\newcommand{\ggraph}{host graph}

\newcommand{\word}{W}

\newcommand{\rev}[1]{#1}
\newcommand{\revc}[1]{#1}
\newcommand{\revp}[1]{#1}

\usepackage{boxedminipage}

\newcommand{\ISIso}{\prb{ISIso}}
\newcommand{\ISIsoR}{\prb{ISIsoR}}

\begin{document}
\title{Changing Induced Subgraph Isomorphisms Under Extended Reconfiguration Rules}

\author[1]{Tatsuhiro Suga
\thanks{suga.tatsuhiro.p5@dc.tohoku.ac.jp}}
\author[1]{Akira Suzuki
\thanks{akira@tohoku.ac.jp}}
\author[1]{Yuma Tamura
\thanks{tamura@tohoku.ac.jp}}
\author[1]{Xiao Zhou
\thanks{zhou@tohoku.ac.jp}}

\affil[1]{Graduate School of Information Sciences, Tohoku University, Sendai, Japan}

\date{}
\maketitle          
\begin{abstract}
In a reconfiguration problem, we are given two feasible solutions of a combinatorial problem and our goal is to determine whether it is possible to reconfigure one into the other, with the steps dictated by specific reconfiguration rules. Traditionally, most studies on reconfiguration problems have focused on rules that allow changing a single element at a time. In contrast, this paper considers scenarios in which $k \ge 2$ elements can be changed simultaneously. We investigate the general reconfiguration \rev{problem} of isomorphisms.
For the \prb{Induced Subgraph Isomorphism Reconfiguration} problem, we show that the problem remains {\PSPACE}-complete even under stringent constraints on the pattern graph when $k$ is constant. We then give two meta-theorems applicable when $k$ is slightly less than the number of vertices in the pattern graph.
In addition, we investigate the complexity of the \prb{Independent Set Reconfiguration} problem, which is a special case of the \prb{Induced Subgraph Isomorphism Reconfiguration} problem.
\end{abstract}

\section{Introduction}
In the field of combinatorial reconfigurations, we focus on the relationships between feasible solutions of combinatorial search problems. This field often addresses reconfiguration problems, which involve determining whether a current feasible solution can be changed to a target feasible solution in a step-by-step manner under a prescribed rule, called a \emph{reconfiguration rule}. 
For example, \rev{in} the \prb{Independent Set Reconfiguration} problem (\prb{ISR} for short) under the token jumping rule, which is one of the most studied combinatorial problems~\cite{BKLMOS21,BB17,DBLP:conf/swat/BonsmaKW14,BFHM21,DDFHIOOUY15,HD05,HU16,IDHPSUU11,KMM12,LM19}, we are given a graph $G$ and its two independent sets $S_s$ and $S_t$.
The goal is to find a sequence $\sigma = \langle S_s = S_0, S_1, ..., S_\ell = S_t \rangle$ of independent sets of $G$ such that for each $i \in \{1, ..., \ell\}$, 
the size of the symmetric difference between $S_{i}$ and $S_{i-1}$ is 2, that is, $|S_{i} \setminus S_{i-1}| = |S_{i-1} \setminus S_{i}| = 1$.
Although the study of reconfiguration problems dates back to classical puzzles like the Tower of Hanoi and the 15-puzzle, research in this area has rapidly advanced since Ito et al.\ proposed a combinatorial reconfiguration framework in 2011~\cite{IDHPSUU11}. Various results on reconfiguration problems are summarized in several survey papers~\cite{DBLP:books/cu/p/Heuvel13,DBLP:journals/algorithms/Nishimura18}.

While reconfiguration problems have attracted attention in theoretical computer science, due to their fundamental properties, they are also well-studied for modeling continuously operating services that cannot afford downtime, such as infrastructure and monitoring systems. For instance, in an electrical power distribution network, it is not practical to halt the power supply to modify the switch configuration. Instead, it is essential to change the current configuration to the desired one step-by-step. To apply the theoretical insights developed in the field of combinatorial reconfigurations to real-world scenarios, various solvers for different reconfiguration problems have been developed. In particular, solvers for \rev{\prb{ISR}} have been created based on AI planning methods~\cite{DBLP:conf/ecai/ChristenE0MPPSS23}, answer set programming~\cite{DBLP:conf/walcom/YamadaBISU24}, zero-suppressed binary decision diagram (ZDD)~\cite{DBLP:conf/cpaior/ItoKNSSTT23}, and bounded model checking~\cite{DBLP:conf/ictai/TodaIKSST23}. Notably, all the above studies were published in 2023 and 2024. 
Thus, combinatorial reconfigurations are progressing from theoretical research to practical applications.

However, the existing reconfiguration rules face certain issues when reconfiguration problems are applied to real-world situations. Suppose that we are required to change the current configuration to perform maintenance on an electrical power distribution network, but we \rev{are} aware that there is no way to reach the desired configuration under the existing reconfiguration rule. It makes no sense to abandon maintenance for that reason. We must somehow find a way to change to the desired configuration.

To resolve this, this paper deals with extended reconfiguration rules. Most of the reconfiguration problems in previous work iteratively change \emph{one} element in each step. We extend this rule and allow the simultaneous change of \emph{two or more} elements in each step.
In reconfiguration problems involving changes to a vertex subset in a graph, we view such a vertex subset as a set of tokens. The token sliding rule ($\TS$) and the token jumping rule ($\TJ$) have been addressed in the literature. For each step, only one token can slide along an edge under $\TS$, and it can jump to any vertex under $\TJ$.
We consider two reconfiguration rules $\kTS$ and $\kTJ$, which are generalizations of $\TS$ and $\TJ$, respectively. Each of them allows at most $k$ tokens to be moved simultaneously.
These reconfiguration rules may provide a way to reach the desired solution, even if the original rule could not (see also \Cref{fig:enter-labelkTJ}).

\begin{figure}
    \centering
    \includegraphics[width=0.5\linewidth]{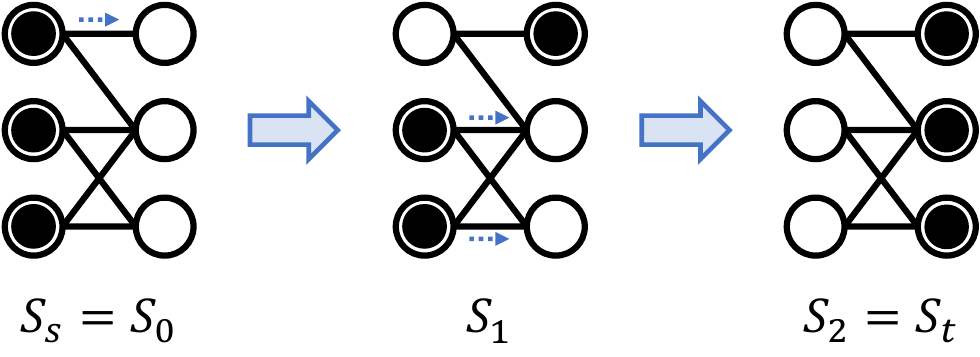}
    \caption{\rev{A sequence $\langle S_s=S_0, S_1, S_2=S_t \rangle$ of independent sets within the same graph cannot be achieved under either $\TS$ or $\TJ$, but it can be achieved under both $\numTS{2}$ and $\numTJ{2}$.} Tokens corresponding to the independent sets are marked in black.}
    \label{fig:enter-labelkTJ}
\end{figure}

\subsection{Our Contribution}
This paper analyzes the computational complexity of a comprehensive reconfiguration problem, namely the \prb{Induced Subgraph Isomorphism Reconfiguration} problem ({\ISIsoR} for short), under the extended rules.
Let $\mathscr{G}$ and $\mathscr{H}$ be graph classes. In {\ISIsoR}, we are given two graphs $G\in \mathscr{G}$ and $H\in \mathscr{H}$, called a \emph{\ggraph} and a \emph{\hgraph}, respectively, two vertex subsets $S_s$ and $S_t$ of $G$, each of which induces the subgraph isomorphic to $H$, and a reconfiguration rule $\Rule$. The objective is to determine whether there is a sequence $\sigma = \langle S_s = S_0, S_1, ..., S_\ell = S_t \rangle$ of vertex subsets of $G$ that satisfies the following two conditions: $S_{i}$ induces the subgraph isomorphic to $H$ for each $i \in \{0, ..., \ell\}$; and $S_{i-1}$ and $S_i$ follow the rule $\Rule$ for each $i \in \{1, ..., \ell\}$.
Here, \prb{ISR} can be viewed as a special case of {\ISIsoR} where $H$ is a null graph (i.e., a graph with no edges).

We first show that, for any additive graph class $\mathscr{H}$, there exists a fixed positive integer $k_{\mathscr{H}}$ such that, for any integer $k\geq k_{\mathscr{H}}$, {\ISIsoR} under $\Rule\in \{\kTS,\kTJ\}$ is \PSPACE-complete for $\mathscr{G}$ and $\mathscr{H}$.
Furthermore, the problem remains {\PSPACE}-complete even when $\mathscr{G}$ is the class of perfect graphs. In contrast, if $k=|V(H)|$ for a {\hgraph} $H\in \mathscr{H}$, {\ISIsoR} under $\kTJ$ is trivially solvable in linear time. We thus focus on the case where $k$ is slightly smaller than $|V(H)|$. To be precise, we consider a parameter $\mu = |V(H)|- k$, assuming $k < |V(H)|$.
Such a parameter, so-called a \emph{{\rm (}below{\rm )} guaranteed value}, was introduced by Mahajan and Raman~\cite{DBLP:journals/jal/MahajanR99} and has been studied in the literature.
For {\ISIsoR} under $\kTJ$, we give negative and positive meta-theorems when parameterized by $\mu = |V(H)|- k$ for a {\hgraph} $H\in \mathscr{H}$.
The negative result is that, for the graph class $\mathscr{G}$ of general graphs and any additive graph class $\mathscr{H}$, if the \prb{Induced Subgraph Isomorphism} problem ({\ISIso} for short) is {\NP}-complete for $\mathscr{G}$ and $\mathscr{H}$, then {\ISIsoR} under $\kTJ$ is also {\NP}-complete for $\mathscr{G}$ and $\mathscr{H}$ when $\mu$ is any fixed positive integer. The positive result is that, if {\ISIso} is solvable in polynomial time for a hereditary graph class $\mathscr{G}$ and a finite assorted graph class $\mathscr{H}$, then {\ISIsoR} under $\kTJ$ is in $\XP$ for $\mathscr{G}$ and $\mathscr{H}$ when parameterized by $\mu$.
As we will see later, a finite assorted graph class $\mathscr{H}$ is also additive. Combining the two meta-theorems, {\ISIsoR} under $\kTJ$ takes over the (in)tractability of {\ISIso} for any finite assorted graph class $\mathscr{H}$.

The results for {\ISIsoR} provide an interesting complexity change of \prb{ISR} under $\kTJ$. Let $I$ be an initial independent set of a given graph $G$. 
\prb{ISR} under $\kTJ$ is {\PSPACE}-complete when $k=\Theta(1)$; {\NP}-complete when $k=|I|-\Theta(1)$; and solvable in polynomial time when $k=|I|$. 
\rev{Although \prb{ISR} under $\kTJ$ on perfect graphs remains {\PSPACE}-complete, we show the problem on perfect graphs is in {\XP} when parameterized by $\mu=|I|-k$.}
The {\XP} algorithm is the best possible under a reasonable assumption Gap-ETH. 
We also study the complexity of \prb{ISR} under $\kTS$. We show that \prb{ISR} under $\kTS$ is essentially equivalent to \prb{ISR} under $\TS$ if a given graph is even-hole-free. This result indicates that, for any integer $k$ with $1\leq k \leq |I|$, 
\prb{ISR} under $\kTS$ is {\PSPACE}-complete even for split graphs and solvable in polynomial time for interval graphs.
\revc{In contrast to \prb{ISR} under $\kTJ$, it is unlikely that an XP algorithm exists when parameterized by $\mu = |I|-k$ for \prb{ISR} under $\kTS$ on split graphs, a subclass of perfect graphs.}

\revp{Proofs marked $(\ast)$ are postponed to Appendices. }

\subsection{Related Work}
Only a few studies have addressed extended reconfiguration rules.
Several \rev{works} have explored the \prb{Shortest Path Reconfiguration} problem under the rule that allows at most $k$ vertices of a shortest path to change simultaneously. This problem is known to be {\PSPACE}-complete when $k$ is constant~\cite{B13,DBLP:conf/walcom/DomonST024,GJKL21,KMM11}, while there is an {\FPT} algorithm when parameterized by $\mu=n/2-k\geq 0$, where $n$ is the number of vertices of an input graph~\cite{DBLP:conf/walcom/DomonST024}.
Additionally, the research of the \prb{ISR} problem has been conducted on how many vertices need to be changed simultaneously to ensure that any independent sets are reachable~\cite{DBLP:journals/dam/BergJM18}.

\prb{ISR} can be viewed as a special case of {\ISIsoR} where $H$ is a null graph. It is known that \prb{ISR} under $\Rule\in\{\TS,\TJ\}$
is {\PSPACE}-complete in general~\cite{HD05,IDHPSUU11,KMM12}, while the problem is solvable in polynomial time in several restricted cases~\cite{BB17,B16,DBLP:conf/swat/BonsmaKW14,BFHM21,DDFHIOOUY15,HU16,KMM12}. The problems under $\Rule\in\{\TS,\TJ\}$ in the cases where {\hgraph}s are paths or cycles have also been explored, and {\PSPACE}-completeness results established~\cite{DBLP:journals/tcs/HanakaIMMNSSV20}. In the case where {\hgraph}s are $(i,j)$-biclique, the problem under $\TJ$ is also {\PSPACE}-complete~\cite{DBLP:journals/tcs/HanakaIMMNSSV20}. Another line of research has shown that the problem under $\Rule\in\{\TS,\TJ\}$ remains {\PSPACE}-complete when a pattern graph comes from complete graphs (so-called \prb{Clique Reconfiguration})~\cite{DBLP:journals/dam/ItoOO23}. 
It is important to note that these results for {\ISIsoR} assume that exactly one vertex can be changed at a time.

\section{Preliminaries}\label{sec:Preliminaries}

For a positive integer $i$, we write $[i] = \{1,2,...,i\}$. 
Let $G=(V,E)$ be a finite, simple, and undirected graph with the vertex set $V$ and the edge set $E$. We denote by $V(G)$ and $E(G)$ the vertex set and the edge set of $G$, respectively. For a vertex $v$ of $G$, we denote by $N_G(v)$ the \emph{open neighborhood} of $v$, that is, $N_G(v)=\{u\in V\mid uv\in E\}$ and by $N_G[v]$ the \emph{closed neighborhood}, that is, $N_G[v]=N_G(v)\cup \{v\}$. 
For a vertex set $X\subseteq V(G)$, we define $N_G(X)=\{v\notin X\mid uv\in E(G), u\in X\}$ and $N_G[X]=N_G(X)\cup X$.
For two sets $X$ and $Y$, we denote by $X\bigtriangleup Y$ the symmetric difference between $X$ and $Y$, that is, $(X\setminus Y) \cup (Y\setminus X)$.
A subgraph of $G$ is a graph $G'$ such that $V(G')\subseteq V(G)$ and $E(G') \subseteq E(G)$.
For a subset $S\subseteq V(G)$, we denote by $G[S]$ the subgraph induced by $S$.
The \emph{complement} of $G$ is the graph $\overline{G}=(V,\overline{E})$ such that $uv \in \overline{E}$ if and only if $uv \notin E$ for any pair of vertices $u,v \in V$. 
\revc{A sequence $\langle v_0, v_1, \ldots, v_\ell \rangle_{G}$ of vertices of a graph $G$, where $v_{i-1}v_i \in E$ for each $i \in [\ell]$, is called a \emph{path} if all the vertices are distinct. It is called a \emph{cycle} if $v_0, v_1, \ldots, v_{\ell-1}$ are distinct and $v_0 = v_\ell$. The value $\ell$ is referred to as the \emph{length} of the path or cycle.
}
A graph $G$ is said to be \emph{connected} if there exists a path between $u$ and $v$ for any pair of vertices $u,v \in V$. 
The \emph{diameter} of a connected graph $G$ is the largest length of a shortest path between any two vertices of $G$.
A \emph{component} of $G$ is a maximal connected subgraph of $G$.
Let $\mathscr{C}_G$ denote the set of components of $G$.
For two graphs $G_1$ and $G_2$, we denote by $G_1+G_2$ the \emph{disjoint union} of $G_1$ and $G_2$, that is, $G_1 + G_2 = (V(G_1) \cup V(G_2) , E(G_1) \cup E(G_2))$.
For an integer $t$ and a graph $G$, we denote by $tG$ the disjoint union of $t$ copies of $G$.
For a graph $G$ and a vertex $u\in V(G)$, \emph{duplicating} $u$ is to add a new vertex $v$ and add a new edge $vw$ for each $w \in N_G(u)$.
Similarly, for a graph $G$ and a vertex subset $S\subseteq V(G)$, \emph{duplicating} $S$ is to duplicate each vertex in $S$.
For two vertex sets $A, B \subseteq V(G)$, \emph{joining} $A$ and $B$ is to add a new edge between each pair of vertices $u\in A$ and $v\in B$.
For two graphs $G$ and $F$, \emph{replacing} a vertex $v\in V(G)$ with a graph $F$ is to remove $v$ from $G$, add the vertices and edges in $F$ to $G$, and add new edges so that $N_G(v)=N_G(w)$ for each vertex $w\in V(F)$.
For a positive integer $i$, we denote by $K_i$ the complete graph with $i$ vertices.

A graph class is a (non-empty) family of graphs satisfying a certain property.
A graph class $\mathscr{G}$ is called \emph{hereditary} if any induced subgraph $G'$ of $G$ is in $\mathscr{G}$ for any graph $G \in \mathscr{G}$.
A graph class $\mathscr{G}$ is called \emph{additive} if $G_1 + G_2$ is in $\mathscr{G}$ for any two graphs $G_1, G_2 \in \mathscr{G}$.
We say that a graph class $\mathscr{G}$ is \emph{finite assorted} if there exists a family $\mathscr{F}$ of connected graphs with constant size such that, for any positive integer $\ell \leq|\mathscr{F}|$, any $\ell$ positive integers $t_1,t_2,\cdots,t_\ell$, and any $\ell$ graphs $F_1, F_2,..., F_\ell \in \mathscr{F}$, the graph $t_1F_1+t_2 F_2+\cdots+t_\ell F_\ell$ belongs to $\mathscr{G}$.
In other words, any graph in a finite assorted graph class $\mathscr{G}$ is obtained by the union of graphs in $\mathscr{F}$.
For example, the class of graphs of maximum degree $1$ is finite assorted, because there exists $\mathscr{F} = \{K_1, K_2\}$.

A \emph{hole} is an induced cycle with a length at least four. An \emph{odd hole} (resp.\ \emph{even hole}) is a hole with an odd (resp.\ even) length. A graph $G$ is \emph{odd-hole-free} (resp.\ \emph{even-hole-free}) if $G$ has no odd hole (resp.\ even hole) as an induced subgraph.
A graph $G$ is called a \emph{perfect graph} if both $G$ and $\overline{G}$ are odd-hole-free~\cite{Chudnovsky2002TheSP}.
A graph $G$ is called a \emph{bipartite graph} if $V(G)$ can be partitioned into two independent sets.
It is well known that every bipartite graph is a perfect graph.
 
\subsection{Our problems}
A graph $H$ is \emph{isomorphic} to a graph $G$ if there exists a bijection $\rho$ from $V(H)$ to $V(G)$ such that, for each $u,v \in V(H)$, $uv\in E(H)$ if and only if $\rho(u)\rho(v) \in E(G)$.
A graph $H$ is \emph{induced subgraph isomorphic} to $G$ if $H$ is isomorphic to some induced subgraph of $G$.
A vertex set $V'\subseteq V(G)$ is called an $H$-\emph{induced subgraph isomorphic set} if $H$ is isomorphic to $G[V']$.
Let $\mathscr{G}$ and $\mathscr{H}$ be graph classes.
The \prb{Induced Subgraph Isomorphism} problem ({\ISIso} for short) asks whether, given a \emph{\ggraph} $G\in \mathscr{G}$ and a \emph{\hgraph} $H\in \mathscr{H}$, there is an {\hISIS} of $G$.
\prb{Independent Set} and \prb{Clique}, which are well-known \NP-complete problems~\cite{GJ79}, are special cases of {\ISIso}.
We define a \rev{reconfiguration} variant of {\ISIso}, that is, \prb{Induced Subgraph Isomorphism Reconfiguration} ({\ISIsoR} for short). In the problem, we are given a {\ggraph} $G\in \mathscr{G}$, a {\hgraph} $H\in \mathscr{H}$, two {\hISIS}s $S_s\subseteq V(G)$ and $S_t\subseteq V(G)$, and a \rev{reconfiguration} rule $\Rule$. Then, the problem asks whether there exists a reconfiguration sequence $\sigma=\langle S_s=S_0,S_1,...S_{\ell}=S_t \rangle$ of {\hISIS}s and any two consecutive sets in $\sigma$ follow the reconfiguration rule $\Rule$.

For two vertex subsets $S$ and $S^\prime$ of $G$, the change from $S$ to $S^\prime$ can be viewed as the move of tokens on the vertices in $S \bigtriangleup S^\prime$.
In this paper, we consider the following extended reconfiguration rules for {\ISIsoR}.
\begin{itemize}
    \item \textbf{$k$-Token Jumping}~($\kTJ$): at most $k$ tokens on vertices in $G$ can be moved simultaneously.
    \item \textbf{$k$-Token Sliding}~($\kTS$): at most $k$ tokens on vertices in $G$ can be moved simultaneously and each token can be moved to an adjacent vertex in $G$.
\end{itemize}
For two {\hISIS}s $S$ and $S'$, if we can change from $S$ to $S'$ by applying $\kTJ$ (resp. $\kTS$) at a time, then we call that $S$ and $S'$ are \emph{adjacent} under $\kTJ$ (resp. $\kTS$).
If there exists a reconfiguration sequence between $S$ and $S'$ under $\kTJ$ (resp. $\kTS$), then we call that $S$ and $S'$ are \emph{reconfigurable} under $\kTJ$ (resp. $\kTS$).
Note that, when $k=1$, $\kTJ$ and $\kTS$ are equivalent to $\TJ$ and $\TS$, which are fundamental and broadly studied reconfiguration rules in the field of combinatorial reconfigurations.

{\ISIsoR} under $\kTJ$ is trivial if $k$ is the number of tokens, as we can move all tokens on an {\hISIS}.

\section{Induced Subgraph Isomorphism Reconfiguration} \label{sec:ISIR}

\subsection{PSPACE-completeness}\label{ISIRPSPACEcomp}
In this subsection, we show that for a sufficiently large positive constant $k$, {\ISIsoR} under $\Rule\in \{\kTS,\kTJ\}$ is \PSPACE-complete even for any additive graph class $\mathscr{H}$.
The \PSPACE-hardness of {\ISIsoR} under $\Rule\in \{\kTS,\kTJ\}$ is provided by a polynomial-time reduction from the \prb{$\word$-Word Reachability} problem, which is known to be \PSPACE-complete~\cite{W18}. 
\rev{Our reduction is inspired by the work of Kamiński et al., which proved the \PSPACE-hardness of \prb{ISR} for perfect graphs~\cite{KMM12}}.
Given a pair $\word=(\Sigma, A)$, where $\Sigma$ is a set of symbols and $A\subseteq \Sigma^2$ is a binary relation between symbols, a string over $\Sigma$ is a \emph{$\word$-word} if every two consecutive symbols are in $A$. 
The \prb{$\word$-Word Reachability} problem asks whether two given $\word$-words $w_s,w_t$ of equal length can be transformed into one another by changing one symbol at a time so that all intermediate strings are also $\word$-words.

\begin{theorem}\label{sub-reachPSPACE}
    Let $\mathscr{G}$ be the graph class of general graphs, and let $\mathscr{H}$ be any additive graph class. Then there exists a fixed positive integer $k_{\mathscr{H}}$ depending on $\mathscr{H}$ such that, for any integer $k\geq k_{\mathscr{H}}$, {\ISIsoR} under $\Rule\in \{\kTS,\kTJ\}$ is \PSPACE-complete for $\mathscr{G}$ and $\mathscr{H}$.
\end{theorem}

\begin{proof}[Proof~(Sketch)]
    We only give the construction of our instance.
    \revp{A formal proof of \Cref{sub-reachPSPACE} will be given in \Cref{appendixISIR}.}
    
    Consider an instance $(w_s,w_t)$ of the \prb{$\word$-Word Reachability} problem where $W=(\Sigma,A)$. 
    We will construct an instance $(G,H,S_s,S_t,\Rule)$ of {\ISIsoR} under $\Rule\in \{\kTS,\kTJ\}$~(see also \Cref{fig:enter-labelPSPACE}). Let $\Sigma=\{\sigma_1,\sigma_2,...,\sigma_{|\Sigma|}\}$ and $n=|w_s|=|w_t|$. Let $F\in \mathscr{H}$ be a graph with the smallest size in $\mathscr{H}$. Note that $F$ is independent of the instance of {\ISIsoR} and hence fixed. We set $k$ to an arbitrary positive integer at least $k_{\mathscr{H}} = 2|V(F)|$. Let $m$ be a positive integer such that $2^m|V(F)|\leq k < 2^{m+1}|V(F)|$ and $t=2^m$.
    We construct a graph $G'$ as follows. The vertex set $V(G')$ is the union of $n$ vertex sets $L_1, L_2,..., L_n$ such that each vertex set is with size $|\Sigma|$ and is a clique of $G'$. We call each of the vertex sets $L_1, L_2,..., L_n$ a \emph{layer}.
    \revp{Each vertex in a layer corresponds to a symbol in $\Sigma$. We denote by $v_i$ the vertex corresponding to the symbol $\sigma_i$ for some $i\in [|\Sigma|]$.}
    For a positive integer $i\in [n-1]$ and two positive integers $j,j'\in[|\Sigma|]$, two vertices $v_j\in L_{i}$ and $v_{j'}\in L_{i+1}$ are joined by an edge if and only if $(\sigma_j,\sigma_{j'})\notin A$. 
    Then, $G$ is a graph obtained from $G'$ by replacing $v_j\in L_i$ with a graph $tF_i^j$, where $F_i^j$ is isomorphic to $F$ for each pair of $i\in [n]$ and $j\in[|\Sigma|]$.
    This completes the construction of $G$.
    Let $H=ntF$. Note that $H \in \mathscr{H}$ because $F \in \mathscr{H}$ holds and $\mathscr{H}$ is additive. For a $\word$-word $w$ and the graph $G$, we associate $w$ with an $H$-induced subgraph isomorphic set $S_w$ in $G$ as follows.
    Denote by $w[i]$ the $i$-th symbol of $w$.
    For each $i\in[n]$, we add tokens on a vertex set $V(tF_i^j)$ if $w[i]=\sigma_j$ for $j\in[|\Sigma|]$. Lastly, let $S_s=S_{w_s}$ and $S_t=S_{w_t}$. 

    \revc{Lastly, we can show that $(w_s,w_t)$ is a yes-instance of \prb{$\word$-Word Reachability} if and only if $(G,H,S_s,S_t,\Rule)$ is a yes-instance of {\ISIsoR}.}
    
\end{proof}

\begin{figure}
    \centering
    \includegraphics[width=0.8\linewidth]{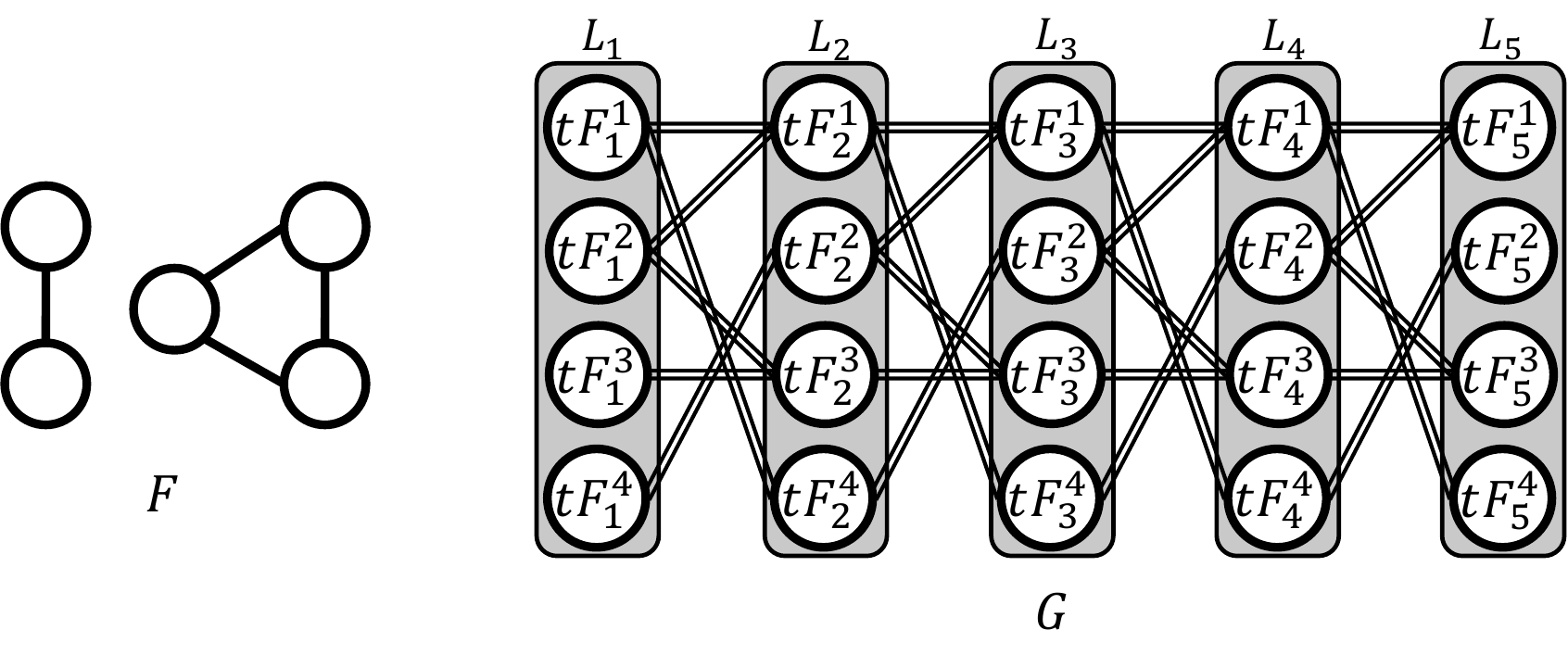}
    \caption{\rev{An example of constructing $G$ from $\word=(\Sigma,A)$ and $(w_s,w_t)$ such that $\Sigma=\{a,b,c,d\}$, $A=\{(a,b),(a,c),(b,b),(b,d),(c,a),(c,b),(c,d),(d,a),(d,c),(d,d)\}$, and $|w_s|=|w_t|=5$. Each node in the layers corresponds to a subgraph isomorphic to $tF$, and each double line represents all possible edges between the two subgraphs.}}
    \label{fig:enter-labelPSPACE}
\end{figure}

From our construction of $G$ in \Cref{sub-reachPSPACE}, we have the following corollary. 

\begin{corollary}[$\ast$]\label{sub-reachPSPACEperfect}
    Let $\mathscr{G}$ be the graph class of perfect graphs, and let $\mathscr{H}$ be any additive graph class that contains some perfect graph. Then there exists a positive integer $k_{\mathscr{H}}$ depending on $\mathscr{H}$ such that, for any integer $k\geq k_{\mathscr{H}}$, {\ISIsoR} is \PSPACE-complete under $\Rule\in \{\kTS,\kTJ\}$ for $\mathscr{G}$ and $\mathscr{H}$.
\end{corollary}

It should be \rev{noted} that many graph classes are additive and usually contain some perfect graph, such as $K_1$. \Cref{sub-reachPSPACEperfect} indicates that, for a sufficiently large constant $k$, {\ISIsoR} is \PSPACE-complete under $\Rule\in \{\kTS,\kTJ\}$ for perfect graphs $\mathscr{G}$ and any ``natural'' graph class $\mathscr{H}$.

\subsection{NP-completeness}\label{subsec:iso_NP-completeness}
Recall that {\ISIsoR} \rev{under $\kTJ$} is trivially solvable in linear time if $k \ge |V(H)|$, where $H$ is a pattern graph.
Assuming $k < |V(H)|$, \Cref{subsec:iso_NP-completeness,subsec:iso_XP} are devoted to {\ISIsoR} parameterized by $\mu =|V(H)|- k\geq1$.
In this subsection, we show the hardness result. 
Before providing it, we here define the notion of a \emph{reconfiguration graph}, which will also be used to show the tractability of {\ISIsoR} for special cases.

Let $(G,H,S_s,S_t,\kTJ)$ be an instance of {\ISIsoR}.
We define the \emph{reconfiguration graph} $\mathcal{C=(V,E)}$ for {\ISIsoR} as follows. 
Each vertex in $\mathcal{V}$ corresponds to an {\hISIS} $S$ of $G$ with size exactly $|V(H)|$. 
We call each vertex of $\mathcal{C}$ a \emph{node} in order to distinguish it from a vertex of $G$.
We denote by $w_S$ a node corresponding to an {\hISIS} $S$.
Then two {\hISIS}s $S$ and $S'$ of $G$ are adjacent under $\kTJ$ if and only if the two corresponding nodes $w_S$ and $w_{S'}$ are joined by an edge in the reconfiguration graph.

\begin{theorem}\label{sub-reach-inNP}
    Let $\mathscr{G}$ be the graph class of general graphs, let $\mathscr{H}$ be any additive graph class, and let $\mu$ be any fixed positive integer. If {\ISIso} is \NP-complete for $\mathscr{G}$ and $\mathscr{H}$, then {\ISIsoR} is \NP-complete under $\kTJ$ for $\mathscr{G}$ and $\mathscr{H}$ when $k=|V(H)|-\mu\geq1$ for a pattern graph $H\in \mathscr{H}$.
\end{theorem}
\begin{proof}[Proof~(Sketch)]
    We give the proof sketch of membership in {\NP} and the construction of our instance for showing {\NP}-hardness.
    \revp{A formal proof of \Cref{sub-reach-inNP} will be given in \Cref{appendixISIR}.}

    \revc{To show that {\ISIsoR} under $\kTJ$ is in {\NP} when $k=|V(H)|-\mu$, we prove the diameter of each component of the reconfiguration graph $C$ is $O(n^\mu)$, which implies that the length of a shortest reconfiguration sequence between any pair of two reconfigurable {\hISIS}s is $O(n^\mu)$.}
    
    We next show that {\ISIsoR} for $\mathscr{G}$ and $\mathscr{H}$ is {\NP}-hard if {\ISIso} for $\mathscr{G}$ and $\mathscr{H}$ is \NP-hard by a polynomial-time reduction from {\ISIso} to {\ISIsoR}. 
    Let $(G',H')$ be an instance of {\ISIso} with $2|V(H')| \geq \mu$.
    Even with this constraint on the instance, {\ISIso} remains \NP-hard because if $2|V(H')| < \mu$, then {\ISIso} is solvable in polynomial time by enumerating all sets of constant size $|V(H')| < \mu/2$.
    We will construct an instance $(G,H,S_s,S_t,\kTJ)$ of {\ISIsoR} under $\kTJ$ from $(G',H')$.

    We construct $G$ from $G'$ as follows (see also \Cref{fig:enter-labelNP}).
    Let $A$, $B$, $X$, and $Y$ be four graphs such that each of them is isomorphic to $2H'$, and let $H^*$ be a graph isomorphic to $H'$.
    Let $G_0$ be a graph such that $V(G_0)=\{g',a,b,h^*,x,y\}$ and $E(G_0)=\{g'a,g'b,ab,ah^*,bh^*,ax,by,xy\}$. Let $G$ be a graph obtained from $G_0$ by replacing $g',a,b,h^*,x$ and $y$ with $G',A,B,H^*,X$, and $Y$, respectively.  
    This completes the construction of $G$.
    Finally, set $H=4H'$, $S_s=V(A) \cup V(Y)$, and $S_t=V(B) \cup V(X)$.
    
    \revc{Lastly, we can prove that $(G',H')$ is a yes-instance of {\ISIso} if and only if $(G,H,S_s,S_t,\kTJ)$ is a yes-instance of {\ISIsoR}. }
\end{proof}

\begin{figure}
    \centering
    \includegraphics[width=0.8\linewidth]{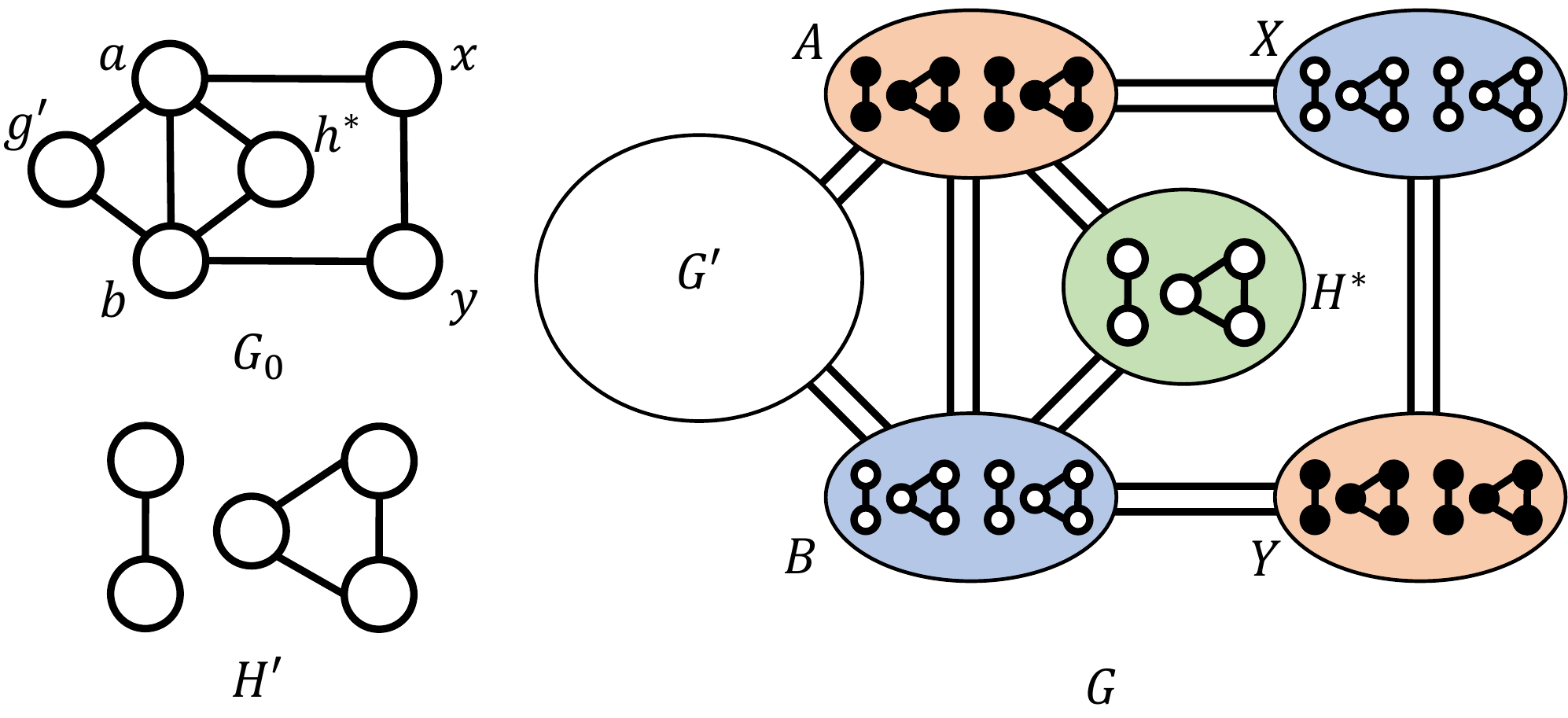}
    \caption{Illustration of our construction of $G$. The double lines represent all possible edges between the two subgraphs. The black tokens are placed on the initial {\hISIS} $S_s$.}
    \label{fig:enter-labelNP}
\end{figure}

Similar to \Cref{sub-reachPSPACEperfect}, we can show that the above hardness holds even for perfect graphs. 

\begin{corollary}[$\ast$]\label{sub-reachNPperfect}
    Let $\mathscr{G}$ and $\mathscr{H}$ be \rev{the graph class of perfect graphs and the additive graph subclass of perfect graphs}, and let $\mu$ be any fixed positive integer. If {\ISIso} is \NP-complete for $\mathscr{G}$ and $\mathscr{H}$, then {\ISIsoR} is \NP-complete under $\kTJ$ for $\mathscr{G}$ and $\mathscr{H}$ when   $k=|V(H)|-\mu\geq1$ for a pattern graph $H\in \mathscr{H}$.
\end{corollary}

\subsection{XP algorithms} \label{subsec:iso_XP}

\begin{theorem} \label{the:ISIR_XP}
    Let $\mathscr{G}$ be any hereditary graph class, and let $\mathscr{H}$ be any finite assorted graph class. If {\ISIso} can be solved in polynomial time for $\mathscr{G}$ and $\mathscr{H}$, then {\ISIsoR} under $\kTJ$ is in {\XP} for $\mathscr{G}$ and $\mathscr{H}$ when parameterized by $\mu=|V(H)|-k\geq 1$ for a pattern graph $H\in \mathscr{H}$.
\end{theorem}

    We give an {\XP} algorithm that solves {\ISIsoR}.
    To this end, we will construct a \emph{clique-compressed reconfiguration graph} $\mathcal{C'=(V',E')}$, which is essentially equivalent to a reconfiguration graph $\mathcal{C=(V,E)}$ but has fewer nodes than $\mathcal{C}$. We call each vertex of $\mathcal{C'}$ a  \emph{clique-node} in order to distinguish it from a node of $\mathcal{C}$ and a vertex of $G$. We denote by $(G,H,S_s,S_t,\kTJ)$ an instance of {\ISIsoR} under $\kTJ$.

    We first define the clique-compressed reconfiguration graph $\mathcal{C'=(V',E')}$ as follows. 
    Clique-nodes in $\mathcal{V'}$ are in one-to-one correspondence with vertex sets of $V(G)$ of size exactly $\mu$. We denote by $w'_T$ the clique-node that is assigned a vertex set $T\subseteq V(G)$. Then two clique-nodes $w'_T$ and $w'_{T'}$ in $\mathcal{C'}$ are joined by an edge if and only if there exists an {\hISIS} $S$ such that $T\cup T'\subseteq S$.

    \begin{proposition}[$\ast$]\label{claim:clique_nodes}
        Let $w_{S_s}$ and $w_{S_t}$ be the two nodes in $\mathcal{C}$ corresponding to $S_s$ and $S_t$, respectively. There is a path in $\mathcal{C}$ connecting $w_{S_s}$ and $w_{S_t}$ if and only if there is a path in $\mathcal{C'}$ connecting two clique-nodes $w'_{I}$ and $w'_{J}$ in $\mathcal{C}'$ for any $I\subseteq S_s$ and $J\subseteq S_t$ with size exactly $\mu$.
    \end{proposition}

    We now explain how to construct the clique-compressed reconfiguration graph $\mathcal{C'=(V',E')}$ without the reconfiguration graph $\mathcal{C}$. 
    
    For each vertex subset $S$ of $G$ with size exactly $\mu$, we create a clique-node in $\mathcal{C'}$ to which $S$ is assigned.
    We can construct the clique-node set $\mathcal{V'}$ in time $O(n^\mu)$ because $|\mathcal{V'}|=\tbinom{n}{\mu}=O(n^\mu)$\revc{, where $n=|V(G)|$}.

    We then construct the edge set $\mathcal{E'}$ of $\mathcal{C'}$. Consider two clique-nodes $w'_{A}$ and $w'_{B}$, where $A$ and $B$ are vertex sets of $G$ with $|A|=|B|=\mu$.
    Recall that the two clique-nodes $w'_{A}$ and $w'_{B}$  of  $\mathcal{C'}$ are joined by an edge if and only if there exists an {\hISIS} $S\subseteq V(G)$ such that $A\cup B\subseteq S$.
    In what follows, we show how to check whether such \rev{a} vertex subset $S$ exists.
    Let $C = A\cup B$. 
    Since $\mathscr{H}$ is a finite assorted graph class, $H \in \mathscr{H}$ can be denoted as $H=t_1F_1+t_2 F_2+\cdots+t_\ell F_\ell$ for some $\ell$ positive integers $t_1,t_2,\cdots,t_\ell$ and some graphs $F_1, F_2,..., F_\ell$ with constant sizes.
    For each vertex $v\in C$, we assign a graph $F_i$ for some $i\in [\ell]$ and we pick a set $W_v$ of $|F_i|$ vertices in $G$ such that $v \in W_v$ as a candidate for an {\gISIS{F_i}} in $G$. For $v_1,v_2,...,v_{|C|}\in C$, let $W=W_{v_1}\cup W_{v_2}\cup\cdots\cup W_{v_{|C|}}$. 
    Note that $A \cup B \subseteq W$.
    If $G[W]$ is not isomorphic to any components of $H$, then we abort the current choice and consider the next choice of $F_i$ and $W$.
    Suppose that $G[W]$ is isomorphic to a graph consisting of some components $F$ of $H$. Then we denote by $G'$ the graph obtained from $G$ by removing all vertices in $N_G[W]$. In addition, we denote by $H'$ the graph obtained from $H$ by removing the vertices of $F$.
    Since $\mathscr{G}$ is a hereditary graph class and $\mathscr{H}$ is a finite assorted graph class, we have $G' \in \mathscr{G}$ and $H' \in \mathscr{H}$.
    We solve the {\ISIso} for the two graphs $G'$ and $H'$ \revc{in time $g(n)$, where $g(n)$ is a polynomial in $n$}. We denote by $(G',H')$ the instance of {\ISIso}. If $(G',H')$ is a yes-instance, then $G'$ has an {\gISIS{H'}} $S_{H'}$. Since $G'$ is a graph obtained from $G$ by removing all vertices in $N_G[W]$, the graph $G$ has an {\hISIS} $W\cup S_{H'}$. Since $A \cup B \subseteq W$, in this case we have found an {\hISIS} $W\cup S_{H'}$ of $G$ such that $A \cup B \subseteq W\cup S_{H'}$. Therefore, the two clique-nodes $w'_{A}$ and $w'_{B}$ are joined by an edge. 
    If $(G',H')$ is a no-instance, we proceed to the next choice of $F_i$ and $W$.
    We consider all possible choices of $F_i$ and $W$.
    If there is no vertex set $S$ such that $A \cup B \subseteq S$, we conclude that $w'_{A}$ and $w'_{B}$ are not joined by an edge.
    In this way, we construct $\mathcal{E'}$ by applying the above procedure for any two clique-nodes. 

    We lastly determine whether $S_s$ and $S_t$ are reconfigurable under $\kTJ$ by checking whether there exists a path between two clique-node $w'_A$ and $w'_B$ for arbitrarily chosen vertex sets $A \subseteq S_s$ and $B \subseteq S_t$ of $G$ with size exactly $\mu$. The correctness follows from \Cref{claim:clique_nodes}. 

    \revc{The total running time of this algorithm is $O(n^\mu+\ell^{2\mu} n^{2\mu (f_{\max}+1)} g(n)+n^{2\mu})$, where $f_{\max}$ denotes the largest number of vertices among the graphs $F_1,F_2,...,F_\ell$, and $g(n)$ is a polynomial in $n$.}
    \revp{The detailed explanation is provided in \Cref{appendixISIR}.}
    This completes the proof of \Cref{the:ISIR_XP}.


\section{Independent Set Reconfiguration}\label{sec:ISR}

As applications of the results in \Cref{sec:ISIR}, we discuss \prb{ISR} under $\kTJ$.
Given a graph $G$ and a positive integer $k$, the \textsc{Independent Set} problem (\prb{IS} for short) asks whether $G$ has an independent set of size at \revc{least} $k$.
\prb{IS} is a well-known {\NP}-complete problem.
Recall that \prb{ISR} is a reconfiguration variant of \prb{IS} and is equivalent to {\ISIsoR} when $\mathscr{H}$ is the class of all null graphs.
Combined with the \PSPACE-completeness under $\TJ$, \Cref{sub-reachPSPACEperfect} asserts the following corollary.

\begin{corollary}
    For any fixed integer $k\geq 1$, ISR under $\kTJ$ is {\PSPACE}-complete for perfect graphs.
\end{corollary}

Let $\mu$ be a fixed positive integer, and let $I$ be an initial independent set of given graph $G$.
When $k=|I|-\mu$, since \prb{IS} is {\NP}-complete for general graphs and the class $\mathscr{H}$ of null graphs is finite assorted (more generally, additive), the \NP-completeness of \prb{ISR} under $\kTJ$ follows from \Cref{sub-reach-inNP}.

 \begin{corollary}
     Let $\mu$ be \rev{any} fixed positive integer. \prb{ISR} under $\kTJ$ is \NP-complete for general graphs when $k=|I|-\mu\geq 1$, where $I$ is an initial independent set of a given graph $G$.
 \end{corollary}

In contrast, it is known that \prb{IS} is solvable in polynomial time for perfect graphs~\cite{GROTSCHEL84ISforPerfect}.
Furthermore, the class of perfect graphs is hereditary.
Therefore, \Cref{the:ISIR_XP} gives the following tractability.

 \begin{corollary} \label{cor:ISR_XP}
    \prb{ISR} under $\kTJ$ is in {\XP} for perfect graphs when parameterized by $\mu = |I| - k\geq 1$, where $I$ is an initial independent set of a given graph $G$.
 \end{corollary}

We here show that, for perfect graphs, the {\XP} algorithm by \Cref{cor:ISR_XP} is the best possible under a reasonable assumption.
This result is based on Gap-ETH \cite{DBLP:journals/eccc/Dinur16,DBLP:conf/icalp/ManurangsiR17}. For a positive integer $q$, in the $q$-SAT problem, we are given a $q$-CNF formula $\phi$ in which each clause consists of at most $q$ literals. Our goal is to decide whether $\phi$ is satisfiable. Max $q$-SAT is a maximization \rev{variant} of $q$-SAT. For a positive integer $q$, the \prb{Max $q$-SAT} problem takes as input a $q$-CNF formula $\phi$ and our goal is to compute the maximum number $\mathsf{SAT}(\phi)$ of clauses in $\phi$ that can be simultaneously satisfied. 
The Gap-ETH can now be expressed in terms of $\mathsf{SAT}$ as follows.

\begin{conjecture}[Gap-Exponential Time Hypothesis (Gap-ETH)~\cite{DBLP:journals/eccc/Dinur16,DBLP:conf/icalp/ManurangsiR17}]
    For some constant $\delta,\epsilon>0$, no algorithm can, given a \prb{3-SAT} formula $\phi$ on $n$ variables and $m=O(n)$ clauses, distinguish between the following cases correctly with probability at least $2/3$ in time $O(2^{\delta n})$: {\rm (}i{\rm )} $\mathsf{SAT}(\phi)=m$ and {\rm (}ii{\rm )} $\mathsf{SAT}(\phi)<(1-\epsilon)m$.
\end{conjecture}

Note that Gap-ETH with $\epsilon = 1/m$ means ETH (Exponential-Time Hypothesis), which is the widely accepted assumption in the field of parameterized complexity.
Under Gap-ETH, we show the following \Cref{noFPTunderGapETH}.

 \begin{theorem}\label{noFPTunderGapETH}
     Assuming Gap-ETH, \prb{ISR} under $\kTJ$ does not admit any
     {\FPT} algorithm for bipartite \rev{graphs} when parameterized by $\mu=|I|-k\geq 1$, where $I$ is an initial independent set of a given graph $G$.
 \end{theorem}

 \begin{proof}
    We perform an \revc{{\FPT}}-reduction from \prb{Maximum Balanced Biclique}. Our reduction is inspired by the \revc{{\FPT}}-reduction of Agrawal et al.\ in~\cite{DBLP:conf/iwpec/AgrawalAD21}, which proved that \prb{ISR} under $\TJ$ for bipartite graphs does not admit an {\FPT} algorithm \rev{parameterized by the number of tokens} under Gap-ETH. 
    In \prb{Maximum Balanced Biclique}, we are given a bipartite graph $G = (A\cup B, E)$ and an integer $\bicliquesize$.
    The goal is to decide whether $G$ contains a complete bipartite subgraph (biclique) with $\bicliquesize$ vertices on each side. \prb{Maximum Balanced Biclique} does not admit an \revc{{\FPT}} algorithm parameterized by $b$ unless Gap-ETH fails~\cite{MBBunderGapETH}. 
    
    We will first construct an instance $(G',I,J,\mu)$ of \prb{ISR} under $\kTJ$ parameterized by $\mu=|I|-k = \bicliquesize$ from an instance $(G,\bicliquesize)$ of \prb{Maximum Balanced Biclique} \revc{(see also \Cref{fig:MMBtoISR})}.
    Let $G^* = (A\cup B, \{ uv \mid u\in A, v\in B, uv \notin E(G) \})$, and let $\Tilde{G}$ be the graph obtained from $G^*$ by duplicating $B$ a total of $c=|A|-\bicliquesize$ times.
    We denote by $B_1,...,B_c$ the vertex subsets obtained by duplicating $B$ and denote $B_0 = B$.
    Let $B' = B_0\cup B_1\cup \cdots \cup B_c$. 
    We perform the following operations on $\Tilde{G}$: add a new biclique with two independent sets $S$ and $T$, where $|S|=|T|=(c+2)\bicliquesize$; join $S$ and $B'$; and join $T$ and $A$.
    We denote by $G'$ the graph obtained as above.
    Let $I=S$ and $J=T$, then we have $|I|=|J|=(c+2)\bicliquesize$.

    To complete our \revc{{\FPT}}-reduction, we now show that $(G,\bicliquesize)$ is a yes-instance if and only if $(G',I,J,\mu)$ is a yes-instance.
    
    Suppose that $G$ has a biclique with $\bicliquesize$ vertices on each side. Then, we construct a reconfiguration sequence between $I$ and $J$ as follows. First, move $\mu = b$ tokens from $S$ to $A$. Note that this move is successful because $k = |I| - \mu = (c+2)b - b = (c+1)b > \mu$ and there is no edge between any pair of a vertex in $S$ and a vertex in $A$. Second, move $(c+1)b = k$ remaining tokens on $S$ to $B'$ so that, for each $i \in [c] \cup \{0\}$, a set of vertices in $A\cup B_i$ on which the tokens are placed is an independent set. This can be done because a biclique with $b$ vertices on each side of $G$ now corresponds to the independent set of $\Tilde{G}[A \cup B_i]$ for every $i \in [c] \cup \{0\}$. Third, move all $\mu$ tokens on $A$ to $T$. Lastly, move all $(c+1)b$ tokens on $B'$ to $T$. Therefore, $I$ and $J$ are reconfigurable.

    Conversely, suppose that there is a reconfiguration sequence $\sigma = \langle I = I_0, I_1, ..., I_\ell = J \rangle$ between $I$ and $J$ under $\kTJ$ with $k=|I|-\mu$. Let $i > 0$ be the minimum integer such that $I_i \cap S = \emptyset$. Considering $I_{i-1}$, we have $I_{i-1} \subseteq S \cup A$ because $I_{i-1} \cap (B' \cup T) = \emptyset$ due to $I_{i-1} \cap S \neq \emptyset$. Recall that $|I_{i-1}| = |I| = (c+2)\bicliquesize$ and $k = |I| - \bicliquesize = (c+1)\bicliquesize$. Since $I_i \cap S = \emptyset$, this implies that $|I_i \cap A| \ge \bicliquesize$ and hence $I_i \cap T = \emptyset$ from the construction of $G'$. Furthermore, since $|A|=c+\bicliquesize$, we have $|I_i \cap B'| \ge |I_i \setminus A| = (c+2)\bicliquesize - (c+\bicliquesize) = (c+1)(\bicliquesize-1)+1$.
    With the pigeonhole principle, this indicates that at least one vertex set among $B_0,B_1,...,B_c$ contains $\bicliquesize$ vertices in $I_i$. Therefore, $\Tilde{G}[A \cup B_i]$ contains an independent set $I_i$ such that $|I_i \cap A| \ge \bicliquesize$ and $|I_i \cap B_j| \ge \bicliquesize$ for some $j \in [c] \cup \{0\}$, that is, $G$ has a biclique with $\bicliquesize$ vertices on each side.
 \end{proof}

 \begin{figure}
     \centering
     \includegraphics[width=0.7\linewidth]{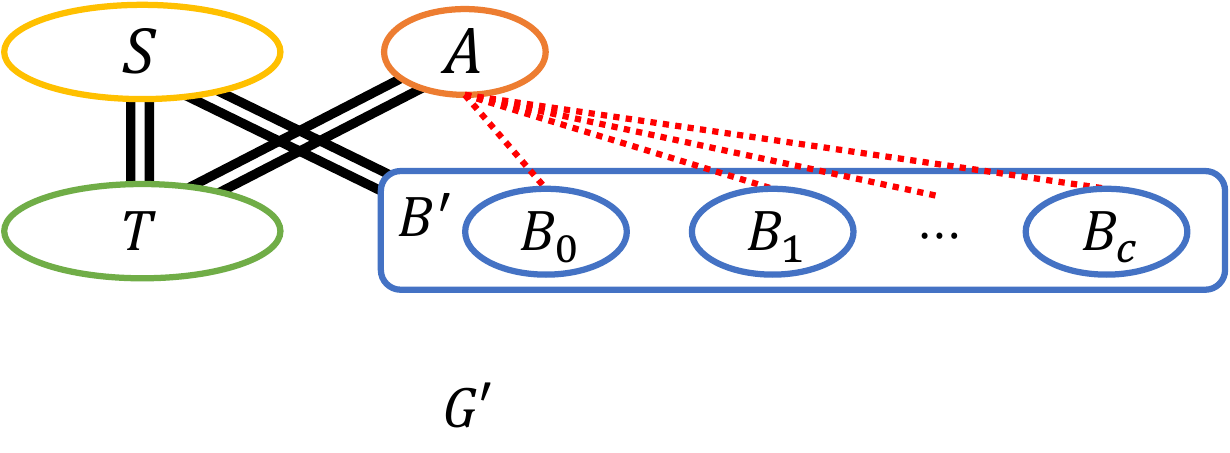}
     \caption{\revc{Illustration of our construction of $G'$. The double lines represent all possible edges between the two vertex sets. The dotted line between $A$ and $B_i$ represents all edges in $G^*$.}}
     \label{fig:MMBtoISR}
 \end{figure}

 We next deal with \prb{ISR} under $\kTS$.
 We show the equivalence between $\TS$ and $\kTS$ for even-hole-free graphs.
 
 \begin{theorem}[$\ast$]\label{TSeqkTS}
     Let $G$ be an even-hole-free graph and let $I$ and $J$ be two independent sets of $G$. Then,  $I$ and $J$ are reconfigurable under $\TS$ if and only if $I$ and $J$ are reconfigurable under $\kTS$. 
 \end{theorem}
 
 \prb{ISR} under $\TS$ is known to \rev{be} \PSPACE-complete even for split graphs~\cite{BKLMOS21} and solvable in polynomial time for interval graphs~\cite{BB17}, \revc{where these graphs are even-hole-free.}
 Therefore, from \Cref{TSeqkTS}, we obtained the following corollary.
 
 \begin{corollary} \label{cor:ISR_kTS}
    For any integer $k$ with $1\leq k \leq |I|$, \prb{ISR} under $\kTS$ is \PSPACE-complete even for split graphs and solvable in polynomial time for interval graphs, where $I$ is an initial independent set of a given instance.
 \end{corollary}

\section{Conclusion} \label{sec:conclusion}
In this paper, to extensively investigate how the extension of reconfiguration rules affects the complexity of reconfiguration problems, we researched a comprehensive reconfiguration problem, namely the \prb{Induced Subgraph Isomorphism Reconfiguration} problem. We presented two meta-theorems that include both negative and positive results for this problem. Our insight would be useful for future research on reconfiguration problems under extended rules.

 \bibliographystyle{plainurl}

%

\newpage
\appendix

\section{Omitted proofs in \Cref{sec:ISIR}}\label{appendixISIR}

\subsection{Proof of \Cref{sub-reachPSPACE}}
\begin{proof}
    Consider an instance $(w_s,w_t)$ of the \prb{$\word$-Word Reachability} problem where $W=(\Sigma,A)$. 
    We will construct an instance $(G,H,S_s,S_t,\Rule)$ of {\ISIsoR} under $\Rule\in \{\kTS,\kTJ\}$~(see also \Cref{fig:enter-labelPSPACE}). Let $\Sigma=\{\sigma_1,\sigma_2,...,\sigma_{|\Sigma|}\}$ and $n=|w_s|=|w_t|$. Let $F\in \mathscr{H}$ be a graph with the smallest size in $\mathscr{H}$. Note that $F$ is independent of the instance of {\ISIsoR} and hence fixed. We set $k$ to an arbitrary positive integer at least $k_{\mathscr{H}} = 2|V(F)|$. Let $m$ be a positive integer such that $2^m|V(F)|\leq k < 2^{m+1}|V(F)|$ and $t=2^m$.
    We construct a graph $G'$ as follows. The vertex set $V(G')$ is the union of $n$ vertex sets $L_1, L_2,..., L_n$ such that each vertex set is with size $|\Sigma|$ and is a clique of $G'$. We call each of the vertex sets $L_1, L_2,..., L_n$ a \emph{layer}. 
    \revp{Each vertex in a layer corresponds to a symbol in $\Sigma$. We denote by $v_i$ the vertex corresponding to the symbol $\sigma_i$ for some $i\in [|\Sigma|]$.}
    For a positive integer $i\in [n-1]$ and two positive integers $j,j'\in[|\Sigma|]$, two vertices $v_j\in L_{i}$ and $v_{j'}\in L_{i+1}$ are joined by an edge if and only if $(\sigma_j,\sigma_{j'})\notin A$. 
    Then, $G$ is a graph obtained from $G'$ by replacing $v_j\in L_i$ with a graph $tF_i^j$, where $F_i^j$ is isomorphic to $F$ for each pair of $i\in [n]$ and $j\in[|\Sigma|]$.
    This completes the construction of $G$.
    Let $H=ntF$. Note that $H \in \mathscr{H}$ because $F \in \mathscr{H}$ holds and $\mathscr{H}$ is additive. For a $\word$-word $w$ and the graph $G$, we associate $w$ with an $H$-induced subgraph isomorphic set $S_w$ in $G$ as follows.
    Denote by $w[i]$ the $i$-th symbol of $w$.
    For each $i\in[n]$, we add tokens on a vertex set $V(tF_i^j)$ if $w[i]=\sigma_j$ for $j\in[|\Sigma|]$. Lastly, let $S_s=S_{w_s}$ and $S_t=S_{w_t}$.

    To complete our polynomial reduction, we prove that $(w_s,w_t)$ is a yes-instance if and only if $(G,H,S_s,S_t,\Rule)$ is a yes-instance.
    First, we will show the only-if direction. Suppose that there exists a reconfiguration sequence $\lambda=\langle w_s=w_0,w_1,...,w_\ell=w_t \rangle$. Let $S_w$ and $S_{w'}$ be two {\hISIS}s of $G$ that are associated with consecutive $\word$-words $w$ and $w'$ in $\lambda$, respectively. Assume that $w$ and $w'$ differ at an $i$-th symbol and let $w[i]=\sigma_j$ and $w'[i]=\sigma_{j'}$. Then, $S_w\setminus S_{w'}=V(tF_i^j)$ and $S_{w'}\setminus S_w=V(tF_i^{j'})$. Since $|V(tF_i^j)|=|V(tF_i^{j'})|=t|V(F)|\leq k$ and any two vertices $v\in V(tF_i^j)$ and $v'\in V(tF_i^{j'})$ are joined by an edge, $S_w$ and $S_{w'}$ are adjacent under $\kTJ$ and $\kTS$. Therefore, there is a reconfiguration sequence $\lambda'=\langle S_s=S_{w_0},S_{w_1},...,S_{w_\ell}=S_t \rangle$.

    Conversely, suppose that there is a reconfiguration sequence $\lambda'=\langle S_s=S_0,S_1,...,S_\ell=S_t \rangle$. We first claim that each vertex set in $\lambda'$ has at most $t|V(F)|$ tokens on vertices in $L_i$ for every $i\in [n]$. For the sake of contradiction, assume that there are more than $t|V(F)|$ tokens on a vertex set $X\subseteq L_i$ for some $i\in [n]$. Then $G[X]$ is connected by our construction of $G$ and we have $|X| > t|V(F)|$. However, any component of $H$ has at most $|V(F)|$ vertices, a contradiction. 
    Similarly, we can claim that all tokens on vertices in $L_i$ for each $i\in [n]$ are placed on vertices of $V(tF_i^j)$ for some $j\in[|\Sigma|]$.
    Thus, for two consecutive {\hISIS}s $S$ and $S'$ in $\lambda'$, we have $S=V(tF_1^{j_1})\cup V(tF_2^{j_2})\cup \cdots \cup V(tF_n^{j_n})$ and $S'=V(tF_1^{j'_1})\cup V(tF_2^{j'_2})\cup \cdots \cup V(tF_n^{j'_n})$, where $j_i, j'_i \in [|\Sigma|]$ for each $i \in [n]$. Then, we convert $S$ and $S'$ to \revp{the corresponding} two strings $w=\sigma_{j_1}\sigma_{j_2}...\sigma_{j_n}$ and $w'=\sigma_{j'_1}\sigma_{j'_2}...\sigma_{j'_n}$, respectively.
    Now we can claim that $w$ is a $\word$-word. Assume that there exists two consecutive symbols $w[i]=\sigma_j$ and $w[i+1]=\sigma_{j'}$ such that $(\sigma_j, \sigma_{j'})\notin A$ for some $i\in[n-1]$. Then, an induced subgraph $G[V(tF_i^j)\cup V(tF_{i+1}^{j'})]$ is connected from our construction. However, since $S$ is an {\hISIS}, any component of $G[S]$ has at most $|V(F)|$ vertices, a contradiction. In the same way, we can claim that $w'$ is also a $\word$-word.
    Furthermore, there exists an integer $a \in [n]$ such that $j_a \neq j'_a$ and $j_b = j'_b$ for every $b \in [n] \setminus \{a\}$; otherwise, $S$ and $S'$ differ by at least $2t|V(F)| > k$ vertices, which contradicts that $S$ and $S'$ are adjacent. It follows that $w$ and $w'$ differ by only one symbol.
    Therefore, $w_s$ and $w_t$ are reconfigurable. 
\end{proof}

\subsection{Proof of \Cref{sub-reachPSPACEperfect}}

\rev{By using the following \Cref{replacingperfect}, which was proved in~\cite{Lovasz06perfecttheorem}, we show \Cref{Gisperfect} to prove \Cref{sub-reachPSPACEperfect}.}

\begin{lemma}[\cite{Lovasz06perfecttheorem}]\label{replacingperfect}
        Let $G_a$ and $F_a$ be two perfect graphs. Then, a graph $G'_a$ obtained from $G_a$ by replacing a vertex $v\in V(G)$ with $F_a$ is also a perfect graph.
\end{lemma}

\begin{lemma}\label{Gisperfect}
    Let $F$ and $G$ be the graphs defined in the proof of \Cref{sub-reachPSPACE}.
    If $F$ is a perfect graph, then $G$ is also a perfect graph.
\end{lemma}

\begin{proof}
    We begin by observing that the graph $G'$ in the proof of \Cref{sub-reachPSPACE} is a perfect graph by following the result of Kaminski et al.~\cite{KMM12}. 
    \revp{For a graph $G = (V, E)$ and a vertex set $V' \subseteq V$, we denote by $G - V'$ the subgraph of $G$ induced by $V \setminus V'$.}
    A \emph{clique cutset} of a graph $G'$ is a clique $K$ in $G'$ such that $G' - K$ is disconnected. For a component $X$ of $G' - K$, if $G'[V(X)\cup K]$ and $G' - V(X)$ are both perfect graphs, then $G'$ is a perfect graph~\cite{KMM12}. Thus, in order to show that $G'$ is a perfect \rev{graph}, it is sufficient to argue that any induced subgraph $G'_s$ of $G'$ without clique cutsets is a perfect graph.
    From our construction of $G'$, $L_2,...,L_{n-1}$ are clique cutsets of $G'$. This implies that any induced subgraph $G'_s$ of $G'$ without clique cutsets is induced by a subset of two consecutive layers $L_i$ and $L_{i+1}$ for $i\in[n-1]$. Thus, since each layer is a clique, $G'_s$ is a co-bipartite graph, that is, the complement of a bipartite graph. It is known that every co-bipartite graph is a perfect graph. We conclude that $G'$ is also a perfect graph.
    
    We next show that the graph $G$ is a perfect graph. 
    Recall that $G$ is obtained from $G'$ by replacing each vertex in $V(G')$ with the disjoint union of some copies of $F$. If $F$ is a perfect graph, the disjoint union of some copies of $F$ is also a perfect graph. Therefore, the proof is completed according to \Cref{replacingperfect}.
\end{proof}

Consider the subclass $\mathscr{H}'$ of $\mathscr{H}$ that consists of all perfect graphs in $\mathscr{H}$. Note that, unless $\mathscr{H}'$ is empty, $\mathscr{H}'$ is additive as the class of perfect graphs is also additive. From \Cref{sub-reachPSPACE}, {\ISIsoR} remains \PSPACE-complete under $\Rule\in \{\kTS,\kTJ\}$ for $\mathscr{G}$ and $\mathscr{H}'$. Combined with \Cref{Gisperfect}, this completes the proof of \Cref{sub-reachPSPACEperfect}.

\subsection{Proof of \Cref{sub-reach-inNP}}

\begin{proof}
    We first show that {\ISIsoR} under $\kTJ$ is in {\NP} when $k=|V(H)|-\mu$ \revp{for any fixed positive integer $\mu$}.
    Let $(G,H,S_s,S_t,\kTJ)$ be an instance of {\ISIsoR}. 
    \revp{Consider a shortest reconfiguration sequence $\sigma = \langle\rho_s=\rho_0,\rho_1,...,\rho_\ell=\rho_t\rangle$ of injections from $V(H)$ to $V(G)$ as the certificate that proves that $(G, H, S_s, S_t, \kTJ)$ is a yes-instance.}
    Note that, if the certificate is valid, the sequence $\sigma$ yields a reconfiguration sequence between $S_s$ and $S_t$. 
    We can check whether each injection in $\sigma$ is an $H$-induced subgraph isomorphism or not in time $O(|V(H)|^2)$. 
    \revp{In addition, we can check whether two images by two consecutive injections in $\sigma$ satisfy $\kTJ$ or not in time $O(|V(H)|)$.}
    Thus, the remaining is to prove that the length of \revp{a shortest} reconfiguration sequence $\sigma$ between two feasible solutions is polynomial. 
    To this end, we show that the diameter of each component in the reconfiguration graph $\mathcal{C}$ is $O(n^\mu)$.
    Since $k=|V(H)|-\mu$, any two feasible solutions with $\mu$ common vertices are adjacent under $\kTJ$. 
    Therefore, nodes in $\mathcal{C}$ corresponding to feasible solutions with $\mu$ common vertices form a clique of $\mathcal{C}$. 
    The reconfiguration graph $\mathcal{C}$ has at most $\tbinom{n}{\mu}=O(n^\mu)$ such cliques $\mathcal{K}$.
    Since each node of $\mathcal{C}$ is contained in at least one clique of $\mathcal{K}$, the vertex set of $\mathcal{C}$ equals the union of all cliques in $\mathcal{K}$.
    A shortest path between any two nodes in a component of $\mathcal{C}$ has at most two nodes of each clique in $\mathcal{K}$.
    Thus, the diameter of each component of $\mathcal{C}$ is $O(n^\mu)$.

    We next show that {\ISIsoR} for $\mathscr{G}$ and $\mathscr{H}$ is {\NP}-hard if {\ISIso} for $\mathscr{G}$ and $\mathscr{H}$ is \NP-hard by a polynomial-time reduction from {\ISIso} to {\ISIsoR}. 
    Let $(G',H')$ be an instance of {\ISIso} with $2|V(H')| \geq \mu$.
    Even with this constraint on the instance, {\ISIso} remains \NP-hard because if $2|V(H')| < \mu$, then {\ISIso} is solvable in polynomial time by enumerating all sets of constant size $|V(H')| < \mu/2$.
    We will construct an instance $(G,H,S_s,S_t,\kTJ)$ of {\ISIsoR} under $\kTJ$ from $(G',H')$.

    We construct $G$ from $G'$ as follows (see also \Cref{fig:enter-labelNP}).
    Let $A$, $B$, $X$, and $Y$ be four graphs such that each of them is isomorphic to $2H'$, and let $H^*$ be a graph isomorphic to $H'$.
    Let $G_0$ be a graph such that $V(G_0)=\{g',a,b,h^*,x,y\}$ and $E(G_0)=\{g'a,g'b,ab,ah^*,bh^*,ax,by,xy\}$. Let $G$ be a graph obtained from $G_0$ by replacing $g',a,b,h^*,x$ and $y$ with $G',A,B,H^*,X$, and $Y$, respectively.  
    This completes the construction of $G$.
    Finally, set $H=4H'$, $S_s=V(A) \cup V(Y)$, and $S_t=V(B) \cup V(X)$.
    \revp{Note that we have $k = |V(H)| - \mu \ge 4|V(H')| - 2|V(H')| = 2|V(H')|$ from the assumption that $2|V(H')| \ge \mu$.}
    
    To complete our polynomial-time reduction, we prove that $(G',H')$ is a yes-instance of {\ISIso} if and only if $(G,H,S_s,S_t,\kTJ)$ is a yes-instance of {\ISIsoR}.

    Suppose that there exists an $H'$-induced subgraph isomorphic set in $G'$. 
    Then there is a reconfiguration sequence between $S_s$ and $S_t$ under $\kTJ$ as follows: First move half of tokens in $A$ on vertices of $G'$ that induce a subgraph isomorphic to $H'$ and the other half in $A$ to $H^*$, which is isomorphic to $H'$; then move all tokens in $Y$ to $X$; lastly move all tokens in $G'$ and $H^*$ to $B$.
    In every reconfiguration step, we move at most $ |V(H)|/2 = 2|V(H')| \le k$ tokens. 
    Therefore, $(G,H,S_s,S_t,\kTJ)$ is a yes-instance.

    Conversely, suppose that there exists a \revp{shortest} reconfiguration sequence $\sigma=\langle S_s=S_0, S_1,...,S_\ell = S_t \rangle$ between $S_s$ and $S_t$ under $\kTJ$. Let $c$ be the number of components of $H'$, then we have $|\mathscr{C}_{H}|=4c$. 
    Furthermore, we have $|S_{\ell-1}\cap (V(B)\cup V(X))| \ge \mu$ because at most $k = |V(H)|-\mu$ tokens can be moved at one reconfiguration step.
    Assume for a contradiction that $S_{\ell-1} \cap V(B) \neq \emptyset$.
    Then, since any vertex in $V(G') \cup V(H^*) \cup V(A) \cup V(Y)$ is adjacent to any vertex in $V(B)$ and $S_{\ell-1}\neq V(B)\cup V(X)$, the graph induced by $S_{\ell-1} \setminus V(X)$ is connected. 
    In other words, the subgraph of $G$ induced by $S_{\ell-1} \cap V(X)$ has exactly $4c-1$ components of $H$, and hence the subgraph contains a subgraph isomorphic to $3H'$.
    However, $X$ is isomorphic to $2H'$, a contradiction.
    Therefore, we have $S_{\ell-1}\cap V(B)=\emptyset$. By the same argument, we can say that $S_{\ell-1}\cap V(A)=\emptyset$. 
    
    We repeatedly perform the following steps (1)--(3).
    \begin{enumerate}[(1)]
        \item Suppose that $S_{\ell-1} \cap V(Y) \neq \emptyset$. Since $S_{\ell-1} \cap V(X) \neq \emptyset$ \revp{due to $|S_{\ell-1}\cap (V(B)\cup V(X))| \ge \mu$ and $S_{\ell-1}\cap V(B)=\emptyset$}, the vertex subset $S_{\ell-1} \cap (V(X) \cup V(Y))$ induces a connected graph $F_{XY}$ of $G$ isomorphic to some component $F_X$ of $X$. We move tokens on the vertices of $F_{XY}$ to the vertices of $F_X$. 
        \revp{This reconfiguration is allowed, since $|V(F_{XY})| \le |V(H')|$ and $k \ge 2|V(H')|$.} 
        Let $i=0$ and $S'_i$ be the vertex subset obtained by this token move. 
        \item Consider a largest component $F$ of $X + H^*$ such that $V(F) \setminus S'_{i} \neq \emptyset$. \revp{Let $f \ge 1$ be the number of components of $H'$ isomorphic to $F$. 
        Then, $H$ has $4f$ components isomorphic to $F$ and $X + H^*$ has exactly $3f$ components isomorphic to $F$.}
        Thus, there exists a component $F_{G'}$ of the subgraph induced by $S'_{i} \cap V(G')$ isomorphic to $F$. We swap tokens on the vertices of $F$ and tokens on the vertices of $F_{G'}$, and then let $S'_{i+1}$ the obtained vertex set. We then increment $i$. 
        \item Repeat the second step until $V(X) \cup V(H^*) \subseteq S'_i$.
    \end{enumerate}
    In the first step, since $G[F_{XY}]$ is connected, $S'_0$ is also an $H$-induced subgraph isomorphic set.
    In the second step, observe that if $S'_i$ is an $H$-induced subgraph isomorphic set, then $S'_{i+1}$ is also an $H$-induced subgraph isomorphic set.
    Furthermore, the above steps eventually terminate because the size of $ S'_i \cap (V(X) \cup V(H^*))$ is strictly increasing.
    Let $S'$ be an obtained vertex set after the above steps terminate.
    As $S'$ is isomorphic to $H = 4H'$ and $S' \cap (V(X) \cup V(H^*)) $ is isomorphic to $3H'$, we conclude that $S' \cap V(G')$ is isomorphic to $H'$. Therefore, $(G',H')$ is a yes-instance.
\end{proof}

\subsection{Proof of \Cref{sub-reachNPperfect}}
\rev{To prove \Cref{sub-reachNPperfect}, we show the following lemma.}
\begin{lemma}\label{NPGisperfect}
    Let $H'$, $G'$, and $G$ be the graphs defined in the proof of \Cref{sub-reach-inNP}. If $H'$ and $G'$ are perfect graphs, then $G$ is also a perfect graph.
\end{lemma}
\begin{proof}
    Let $A$, $B$, $X$, and $Y$ be four graphs such that each of them is isomorphic to $2H'$, and let $H^*$ be a graph isomorphic to $H'$.
    Recall that $G$ is obtained from a graph $G_0$ such that $V(G_0)=\{g',a,b,h^*,x,y\}$ and $E(G_0)=\{g'a,g'b,ab,ah^*,bh^*,ax,by,xy\}$ by replacing $g',a,b,h^*,x$, and $y$ with $G',A,B,H^*,X$, and $Y$, respectively. It is easy to verify that $G_0,A,B,H^*,X$, and $Y$ are perfect graphs. Therefore, from \Cref{replacingperfect}, $G$ is a perfect graph.
\end{proof}
\Cref{NPGisperfect} completes the proof of \Cref{sub-reachNPperfect}.

\subsection{Proof of \Cref{claim:clique_nodes}}

\begin{proof}
    First, we show the only-if direction. Suppose that there exists a path $P_{\mathcal{C}}=\langle w_{S_s}=w_{S_0},w_{S_1},...,w_{S_d}=w_{S_t} \rangle_{\mathcal{C}}$ in $\mathcal{C}$ connecting $w_{S_s}$ and $w_{S_t}$, where $d$ is the length of $P_{\mathcal{C}}$. Consider two vertex sets $S_{i-1}$ and $S_i$ for $i\in [d]$. Since $S_{i-1}$ and $S_i$ are adjacent, we have $|S_{i-1} \setminus S_i| \le k = |V(H)| - \mu$, which implies that $|S_{i-1}\cap S_i| \ge \mu$. Let $T_i$ be a vertex set such that $T_i\subseteq S_{i-1}\cap S_i$ with size exactly $\mu$ for each $i\in [d]$. For each $i \in [d-1]$, since $T_{i}\cup T_{i+1}\subseteq S_{i}$, the two corresponding clique-nodes $w'_{T_{i}}$ and $w'_{T_{i+1}}$ are joined by an edge in $\mathcal{C'}$. Thus, there exists a path $P_{\mathcal{C'}}$ in $\mathcal{C'}$ connecting $w'_{T_1}$ and $w'_{T_d}$, where $T_1\subseteq S_0 = S_s$ and $w'_{T_d}\subseteq S_d = S_t$. Furthermore, $w'_{T_1}$ and $w'_{I}$ with $I\neq T_1$ are joined by an edge, and $w'_{T_d}$ and $w'_{J}$ \revp{with} $J\neq T_d$ are joined by an edge. Therefore, there is a path connecting two clique-nodes $w'_{I}$ and $w'_{J}$ for any $I\subseteq S_s$ and $J\subseteq S_t$ with size exactly $\mu$.

    Then, we show the if direction. Suppose that there is a path $P_{\mathcal{C'}}=\langle w'_I=w'_{I_0},w'_{I_1},...,w'_{I_{d'}}=w'_{J} \rangle_{\mathcal{C'}}$ in $\mathcal{C'}$ connecting two clique-nodes $w'_{I}$ and $w'_{J}$ for any $I\subseteq S_s$ and $J\subseteq S_t$ with size exactly $\mu$. For each $i\in[d']$, let $S'_i$ be an {\hISIS} such that $I_{i-1}\cup I_i\subseteq S'_i$. For each $i\in[d'-1]$, since $|S'_{i}\cap S'_{i+1}|\geq|I_{i}|=\mu$, two nodes $w_{S'_{i}}$ and $w_{S'_{i+1}}$ are joined by an edge. Thus, there exists a path $P^*_{\mathcal{C'}}$ connecting $w_{S'_1}$ and $w_{S'_{d'}}$. $S'_1$ and $S_s$ contain at least $|I_0|=\mu$ common vertices, and $S'_{d'}$ and $S_t$ contain at least $|I_{d'}|=\mu$ common vertices. Therefore, there is a path connecting $w_{S_s}$ and $w_{S_t}$ in $\mathcal{C}$.
\end{proof}

    \subsection{The running time of our XP algorithm in \Cref{subsec:iso_XP}}
    We bound the running time of our algorithm.
    Let $n = |V(G)|$.
    Suppose that there is a function $g$ such that {\ISIso} can be solved in $g(n)$ time for a given graph $G \in \mathcal{G}$.
    The constructed clique-compressed reconfiguration graph has exactly $\tbinom{n}{\mu} = O(n^{\mu})$ clique-nodes.
    For each pair of two distinct clique-nodes, we decide whether there is an edge between them. 
    For each assignment from the vertices of $C = A\cup B$ to $F_1, F_2,..., F_\ell$, the number of possible choices of $W$ are as follows: $\tbinom{n}{|V(F_{i_1})|}\tbinom{n}{|V(F_{i_2})|}\cdots\tbinom{n}{|V(F_{i_{|C|}})|}=O(n^{2\mu f_{\max}})$, where $i_1,i_2,...,i_{|C|}\in [\ell]$ and $f_{\max}$ is the largest number of vertices among the graphs $F_1, F_2,..., F_\ell$.
    Furthermore, there are $O(\ell^{2\mu})$ possible assignments from the vertices of $C = A\cup B$ to $F_1, F_2,..., F_\ell$.
    Therefore, we can construct the edge set with size at most $O(n^{2\mu})$ in time $O(\ell^{2\mu} n^{2\mu (f_{\max}+1)} g(n))$.
    After constructing the clique-compressed reconfiguration graph, by using breadth-first search, we can check whether $(G,H,S_s,S_t,\kTJ)$ of {\ISIsoR} is a yes-instance or not in time $O(|\mathcal{V'}|+|\mathcal{E'}|)=O(n^{2\mu})$.
    The total running time of this algorithm is $O(n^\mu+\ell^{2\mu} n^{2\mu (f_{\max}+1)} g(n)+n^{2\mu})$. Recall that $\ell$ and $f_{\max}$ are constants independent of the instance $(G,H,S_s,S_t,\kTJ)$ of {\ISIsoR}. Therefore, if {\ISIso} for $\mathscr{G}$ and $\mathscr{H}$ can be solved in polynomial time $g(n)$, then {\ISIsoR} for $\mathscr{G}$ and $\mathscr{H}$ is in {\XP} under $\kTJ$ parameterized by $\mu=|V(H)|-k$.

\section{Proof of \Cref{TSeqkTS} in \Cref{sec:ISR}}\label{appendixISR}

     Since $\kTS$ is a generalization of $\TS$, the only if direction is clear.
     Then we prove the if direction. Let us start by showing the important relation between $\TS$ and $\kTS$. 

     \begin{proposition}\label{claimTS=kTS}
        Let $G$ be an even-hole-free graph that has two independent sets $I$ and $J$ with the same size. If $I$ and $J$ are adjacent under $\kTS$, $I$ and $J$ are reconfigurable under $\TS$.
     \end{proposition}
     \begin{proof}
        We will prove this by induction on $t = |I \setminus J| \le k$. The base case $t=0$ is trivial. As the induction step, we show that the claim holds for $t > 0$, assuming that the claim holds for $t - 1$.
        Let $G'$ be the graph induced by $I \bigtriangleup J$. Since $G$ is an even-hole-free graph, $G'$ is also an even-hole-free graph. It is known that if $G'$ is an even-hole-free graph, then there exists a vertex $v\in J\setminus I$ with at most one \rev{neighbor} in $I\setminus J$~\cite{KMM12}. In addition, since $I$ and $J$ are adjacent under $\kTS$, each vertex $v\in J\setminus I$ is adjacent to at least one vertex in $I\setminus J$. Therefore, $v$ has exactly one neighbor $u \in I\setminus J$.
        The vertex set $I' = I\cup \{v\} \setminus \{u\}$ is an independent set of $G$ such that $I$ and $I'$ are adjacent under $\TS$. Considering $I'$ and $J$ are adjacent under $\kTS$ and $|I'\setminus J| = t-1$, from the assumption, $I'$ and $J$ are reconfigurable under $\TS$. In conclusion, $I$ and $J$ are reconfigurable under $\TS$, as claimed.
     \end{proof}
    By using \Cref{claimTS=kTS}, if there exists a reconfiguration sequence $\sigma$ between $I$ and $J$ under $\kTS$, then we can convert $\sigma$ to a reconfiguration sequence between $I$ and $J$ under $\TS$. This completes the proof of \Cref{TSeqkTS}. 
\end{document}